\documentclass[journal]{IEEEtran}
\usepackage[cmex10]{amsmath}
\usepackage{graphicx}
\usepackage{epstopdf}
\usepackage{cases}
\usepackage[tight,footnotesize]{subfigure}
\usepackage{amsthm} 
\usepackage{cite}
\usepackage{amssymb}
\usepackage{algorithm}
\usepackage{algorithmic}
\usepackage{multirow}
\usepackage{amsmath}
\usepackage{xcolor}
\usepackage{CJK}
\usepackage{subeqnarray}
\usepackage{cases}
\usepackage{enumerate}
\usepackage{booktabs}
\usepackage{stfloats}

\newtheorem{theorem}{\hskip\parindent\bf{Theorem}}

\ifCLASSINFOpdf
\else
\fi

\begin{document}

\title{Maximizing Energy Charging for UAV-assisted MEC Systems with SWIPT}

\author{Xiaoyan Hu, \IEEEmembership{Member,~IEEE,} Pengle Wen, \IEEEmembership{Student Member,~IEEE,} \\
Han Xiao, \IEEEmembership{Student Member,~IEEE,} 
Wenjie Wang, \IEEEmembership{Member,~IEEE,} Kai-Kit Wong, \IEEEmembership{Fellow,~IEEE}

\thanks{X. Hu, P. Wen, H. Xiao, and W. Wang, are with the School of Information and Communications Engineering, Xi'an Jiaotong University, Xi'an 710049, China. (email: xiaoyanhu@xjtu.edu.cn, pengle\_wen@stu.xjtu.edu.cn, hanxiaonuli@stu.xjtu.edu.cn,  wjwang@mail.xjtu.edu.cn, zhousu@xjtu.edu.cn).}
	\thanks{K.-K. Wong is with the Department of Electronic and Electrical Engineering, University College London, London WC1E 7JE, U.K. (email: kai-kit.wong@ucl.ac.uk)}
}
\maketitle

\begin{abstract}
A Unmanned aerial vehicle (UAV)-assisted  mobile edge computing (MEC) scheme with simultaneous wireless information and power transfer (SWIPT) is proposed in this paper. Unlike existing MEC-WPT schemes that disregard the downlink period for returning computing results to the ground equipment (GEs), our proposed scheme actively considers and capitalizes on this period. By leveraging the SWIPT technique, the UAV can simultaneously transmit energy and the computing results during the downlink period. In this scheme, our objective is to maximize the remaining energy among all GEs by jointly optimizing computing task scheduling, UAV transmit and receive beamforming, BS receive beamforming, GEs' transmit power and power splitting ratio for information decoding, time scheduling, and UAV trajectory. We propose an alternating optimization algorithm that utilizes the semidefinite relaxation (SDR), singular value decomposition (SVD), and fractional programming (FP) methods to effectively solve the nonconvex problem. Numerous experiments validate the effectiveness of the proposed scheme.
\end{abstract}
\IEEEpeerreviewmaketitle

\vspace{-3mm}
\section{Introduction}\label{sec:Introduction}
The technology of mobile edge computing (MEC) enables users to offload computing tasks to the nearby edge servers for processing, which significantly reduces the computing latency and the energy consumption of the user devices. The practical applications and future development trends of MEC have been extensively studied in \cite{hu2015mobile}.
In general, edge computing servers are fixed on the ground in the traditional MEC systems, potentially resulting in limited service coverage.
Integrating unmanned aerial vehicles (UAVs) with MEC can overcome these limitations, enhancing coverage and improving the efficiency of the MEC system due to their impressive mobility and flexibility. Specifically, in \cite{hu2019uav}, the authors explored a framework for MEC supported by a UAV, where the UAV can act as a computing server to assist ground equipment (GE) in processing computing tasks and serve as a relay to further offload GEs' computation tasks to the base station (BS).

While the MEC technology is capable of effectively processing GEs' computation tasks remotely, it cannot work well in scenarios where the GEs's battery power is insufficient and demand additional energy to sustain normal operations including task offloading.
Hence, leveraging the wireless charging technology into the MEC systems can help address this energy-insufficiency problem \cite{8805125,9312671,hu2018wireless}.
In \cite{8805125}, a UAV-enabled MEC system is explored, where the UAV initially charges the GEs using wireless power transfer (WPT), and then each GE sends its tasks to the UAV for processing. The maximization of the computation energy efficiency for a non-orthogonal multiple access (NOMA)-based WPT-MEC networks is  investigated in \cite{9312671}. Additionally, the authors in \cite{hu2018wireless} examine the minimization of the total transmit energy for BS in WPT-MEC networks.
However, most existing works do not consider the downlink period for returning the calculation results to GEs, which does not align with the practical situations.
In fact, the downlink period should also be taken into consideration, and we can capitalize on the  simultaneous wireless information and power transfer (SWIPT) technology to transmit energy and results simultaneously during this period.
This not only aligns the WPT-MEC systems more closely with the real scenarios but also boosts the overall efficiency of the systems.

Motivated by the above analysis, we establish an optimization problem for a UAV-assisted MEC-SWIPT network considering both the uplink and downlink periods. It aims at maximizing the minimum remaining energy among GEs by jointly designing the computing tasks scheduling, transmit and receive beamforming of the UAV, receive beamforming of the BS, transmit and receive power splitting ratio of the GEs, time scheduling, and UAV trajectory. An alternating optimization algorithm based on the semidefinite relaxation (SDR), singular value decomposition (SVD) and fractional programming (FP) techniques is proposed to solve this problem. The developed scheme closely emulates the UAV-assisted MEC system in real-world scenarios and maximizes the utilization of resources through the incorporation of SWIPT technology.

\vspace{-1mm}
\section{System Model and Problem Formulation}\label{sec:system}
\vspace{-2mm}
\begin{figure}[ht]
  \centering
  \includegraphics[width=0.63\linewidth]{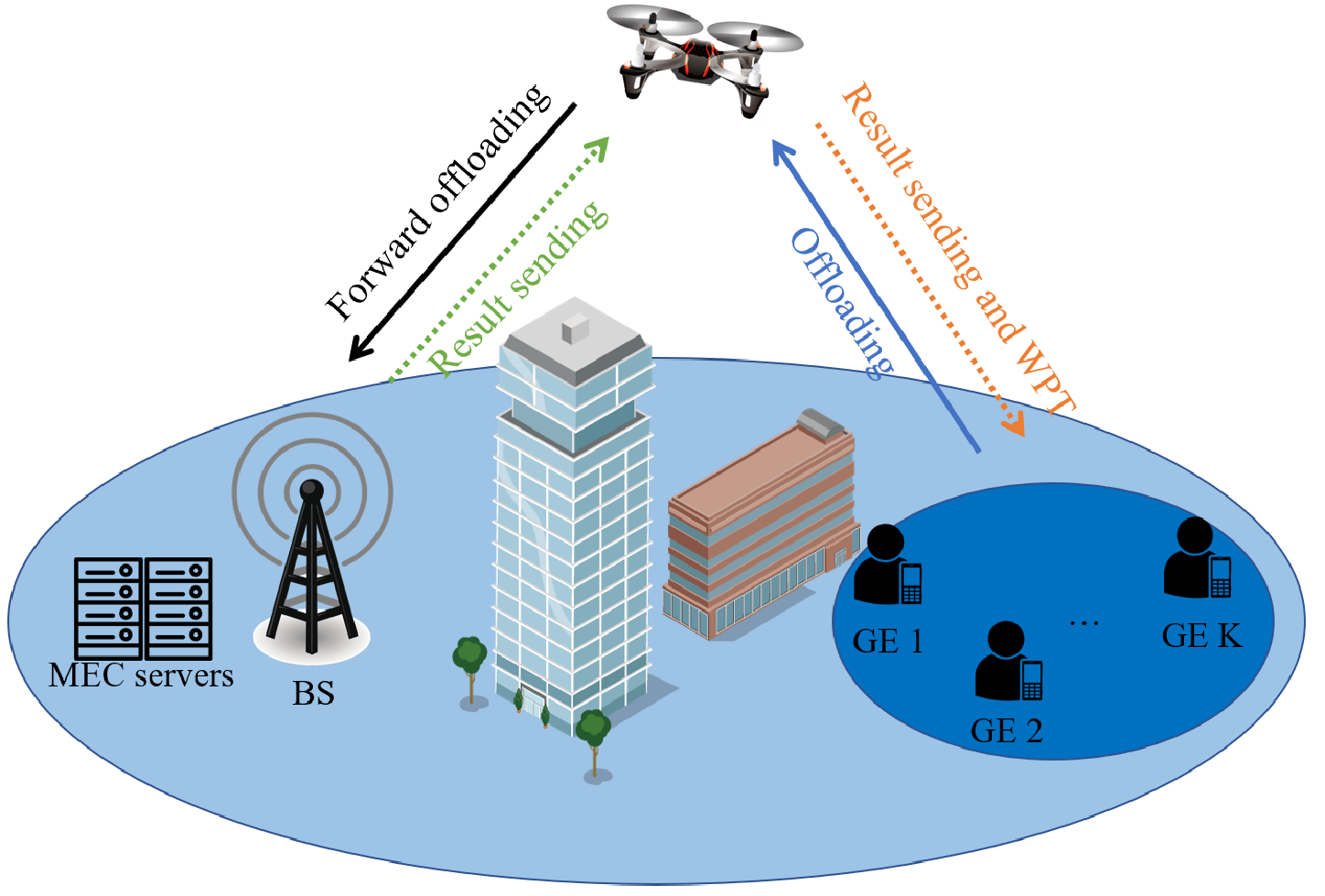}
  \vspace{-1mm}
  \caption{The model of UAV-assisted MEC-SWIPT system.}
  \label{Fig:system1}
\end{figure}
As depicted in Fig. \ref{Fig:system1}, we consider a UAV-assisted MEC network with SWIPT, which consists of a base station (BS) co-located with a MEC server, a UAV, and $K$ GEs denoted as $\mathcal{K}=\left\{1,..., K \right\}$.
Each GE has a computation-intensive task that is bit-wise independent and requires an electrical supply to maintain normal operations. We assume that the direct links between GEs and the BS are blocked by buildings.
The assistant UAV, equipped with $L$ antennas, acts as a relay to send GEs' offloaded tasks to the BS for processing during the uplink period.
Additionally, the SWIPT technology is leveraged at the UAV to transmit the calculation results to GEs and engage in wireless charging simultaneously during the downlink period.
The BS is equipped with a uniform rectangular array of $M=M_{\text{x}}M_{\text{y}}$ antennas, respectively with $M_{\text{x}}$ and $M_{\text{y}}$ elements along the x-direction and y-direction.  

The system is modeled in a three-dimensional (3D) Euclidean coordinate system for all nodes. We divide the time of flight $T$ into $N$ time slots, each slot with the length of $\delta=\frac{T}{N}$, where $\delta$ is sufficiently small such that the UAV's location can be assumed to be unchanged during each slot. Let $\mathcal{N}=\left\{1,\cdots,N \right\}$ denote the set of $N$ time slots. The BS and GE $k\in\mathcal{K}$ are located horizontally at ${\mathbf{s}_\mathrm{b}}= (x_{\text{b}},y_{\text{b}})$ and ${\mathbf{s}_k}= (x_k,y_k)$, with zero vertical coordinates. The UAV is assumed to fly at a fixed altitude $H > 0$ and its horizontal locations at the $n$-th time slot are denoted as $\mathbf{q}\left[ n \right] = (x_{\text{u}}[n],y_{\text{u}}[n])$. The initial and final horizontal locations of the UAV are set as $\mathbf{q_{\text{I}}}=(x_{\text{I}},y_{\text{I}})$ and $\mathbf{q_{\text{F}}}=(x_{\text{F}},y_{\text{F}})$, respectively, and thus the maximum flight speed of the UAV is assumed to be $V_{\text{max}}$. The UAV must satisfy the following mobility constraints
\vspace{-1mm}
\begin{align}
&\lVert \mathbf{q}\left[ n+1 \right] -\mathbf{q}\left[ n \right] \rVert \le \delta V_{\text{max}}, ~{{\forall}n=1,\cdots,N-1},\label{u_move1}\\
&\mathbf{q}\left[ 1 \right]=\mathbf{q}_{\text{I}},\mathbf{q}\left[ N \right]=\mathbf{q}_{\text{F}}. \label{u_move2}
\end{align}

Similar to \cite{xu2022computation}, we adopt the Rician channel to model the GE-UAV links and the UAV-BS link. Therefore, we have
\begin{align}
&\hspace{-2mm}\mathbf{h}_{i}[n]=\sqrt{\frac{\beta _0}{d^2_{i}[n]}}\Big(\sqrt{\frac{\zeta}{1+\zeta}}\mathbf{{h}}^{\rm{LoS}}_{i}[n]+\sqrt{\frac{1}{1+\zeta}}\mathbf{{h}}^{\rm{NLoS}}_{i}[n]\Big),
\end{align}
where $i\in \left\{ \left\{ k,\text{u} \right\} ,\left\{\text{u,b}  \right\} \right\}$ indicates the subscripts of the GE $k$-UAV and the UAV-BS links, $\beta_0$ is the average channel power gain at a reference distance of 1 meter (m), $\zeta$ denotes the Rician factor. Besides, $d_{k,\text{u}}[n]=\sqrt{\lVert \left. \mathbf{q}\left[ n \right] -\mathbf{s}_k \rVert ^2+H^2\right.}$ and $d_{\text{u,b}}[n]=\sqrt{\lVert \left. \mathbf{q}\left[ n \right] -\mathbf{s}_\mathrm{b} \rVert ^2+H^2\right.}$ are the distances from  GE $k$  to the UAV and from the UAV to the BS, respectively.

For the Line of Sight (LoS) component, we have
$\mathbf{{h}}_{k,\text{u}}^{\text{LoS}}[n] = \Big[1, e^{-j\frac{2\pi}{\lambda}d\varPhi_{k,\text{u}}[n]}, \ldots, e^{-j\frac{2\pi}{\lambda}d(L-1)\varPhi_{k,\text{u}}[n]}\Big]^T\in\mathbb{C}^{L\times1}$, where $\lambda$  represents the carrier wavelength, $d$ is the distance between antennas, and $\varPhi_{k,\text{u}}[n] = \frac{x_{\text{u}}[n] - x_k}{\sqrt{\lVert \mathbf{q}[n] - \mathbf{s}_k \rVert^2 + H^2}}$ indicates the cosine of the angle of arrival (AoA) for the signal from GE $k$ to the UAV.
In addition\footnote{{We use the capital letter $\mathbf{{H}}_{\text{u,b}}$  to represent the UAV-BS channel  considering the fact that it is a matrix instead of a vector.}}, {$\mathbf{{H}}_{\text{u,b}}^{\text{LoS}}[n] = \boldsymbol{\phi}_{\text{b,r}}[n] \boldsymbol{\phi}_{\text{u,b}}^H[n]\in\mathbb{C}^{M_{\text{x}}M_{\text{y}}\times L}$}, where
$\boldsymbol{\phi}_{\text{u,b}}[n] = \left[1, e^{-j\frac{2\pi}{\lambda}d\varphi_\text{{ub}}[n]}, \ldots, e^{-j\frac{2\pi}{\lambda}d(L-1)\varphi_\text{{ub}}[n]}\right]^T\in\mathbb{C}^{L\times1}$ denotes the array response {with respect to (w.r.t.)} the angle of departure (AoD) for the signal from the UAV to the BS with $\varphi_\text{{ub}}[n] = \frac{x_{\text{b}} - x_{\text{u}}[n]}{\sqrt{\lVert \mathbf{q}[n] - \mathbf{s}_{\text{b}} \rVert^2 + H^2}}$ being the cosine of the AoD, and
$\boldsymbol{\phi}_{\text{b,r}}[n] = \left[1, e^{-j\frac{2\pi}{\lambda}d\varphi_\text{{br,x}}[n]}, \ldots, e^{-j\frac{2\pi}{\lambda}d(M_{\text{x}}-1)\varphi_\text{{br,x}}[n]}\right]^T \otimes
\left[1, e^{-j\frac{2\pi}{\lambda}d\varphi_\text{{br,y}}[n]}, \ldots, e^{-j\frac{2\pi}{\lambda}d(M_{\text{y}}-1)\varphi_\text{{br,y}}[n]}\right]^T\in\mathbb{C}^{M_{\text{x}}M_{\text{y}}\times1}$ indicates the receive array response at the BS,
with $\varphi_\text{{br,x}}[n] = \sin \varpi[n] \sin \varTheta[n]$ and $\varphi_\text{{br,y}}[n] = \sin \varpi[n] \cos \varTheta[n]$ respectively denoting the vertical and horizontal AoAs of the signals from the UAV to the BS.
Here we have $\sin \varpi[n] = \frac{H}{\sqrt{\lVert \mathbf{q}[n] - \mathbf{s}_{\text{b}} \rVert^2 + H^2}}$, $\sin \varTheta[n] = \frac{x_{\text{b}} - x_{\text{u}}[n]}{\sqrt{\lVert \mathbf{q}[n] - \mathbf{s}_{\text{b}} \rVert^2}}$, and $\cos \varTheta[n] = \frac{y_{\text{b}} - y_{\text{u}}[n]}{\sqrt{\lVert \mathbf{q}[n] - \mathbf{s}_{\text{b}} \rVert^2}}$.

Without loss of generality, we assume that the Non-LoS (NLoS) components $\mathbf{{h}}_{k,\text{u}}^{\text{NLoS}}[n] \in \mathbb{C}^{L\times 1}$ and ${\mathbf{{H}}_{\text{u,b}}^{\text{NLoS}}[n]} \in \mathbb{C}^{M\times L}$  follow the complex normal distributions of  $\mathcal{CN}(\mathbf{0,I}_L)$ and $\mathcal{CN}(\mathbf{0,I}_{M\times L})$, respectively.
It is  assumed that the channel reciprocity holds for all the uplink and downlink channels considered in this paper.
For simplicity of expression, we define ${\overline{\mathbf{H}}_{{\text{u,b}}}\left[ n \right] =\mathbf{H}_{\text{u,b}}}\left[ n \right]d_{\text{u,b}}\left[ n \right]$ and $\overline{\mathbf{h}}_{k,\text{u}}\left[ n \right] =\mathbf{h}_{k,\text{u}}\left[ n \right]d_{k,\text{u}}\left[ n \right]$.  

Hence, the signal-to-interference-plus-noise ratio (SINR) of GE $k$'s signal recovered at the UAV in time slot $n$ for $k\in\mathcal{K}$ and $n\in\mathcal{N}$ can be expressed as
\begin{align}
    &r_k\left[ n \right] =\frac{\frac{E_k[n]}{t_{\text{o}}[n]}\left|\mathbf{v}^{\text{H}}_{k}\left[ n \right] \mathbf{h}_{k,\text{u}}\left[ n \right]   \right|^2}{\sum_{j\ne k}{\frac{E_j[n]}{t_{\text{o}}[n]}\left|\mathbf{v}^{\text{H}}_{k}\left[ n \right] \mathbf{h}_{j,\text{u}}\left[ n \right]   \right|^2}+\ \lVert \mathbf{v}_{k}\left[ \text{n} \right] \rVert ^2
\sigma ^2},
\end{align}
where $\mathbf{v}_k[n]\in\mathbb{C}^{L\times 1}$ represents the receive beamforming at the UAV for GE $k$, while $E_k[n]$ and $t_{\text{o}}[n]$ respectively denote the transmit energy consumption of GE $k$ and the allocated time for the uplink offloading period at time slot $n$. Additionally, $\sigma^2$ indicates the noise power at the receiver.

The transmission rate of the UAV for uplink task offloading to the BS  at time slot $n$ can be given by
\begin{align}
    &R_{\text{\text{o,u}}}\left[ {n} \right]=B\log\text{det}\left(\mathbf{I}_M+\mathbf{\varTheta}_{\text{\text{o,u}}}[n]\right),
\end{align}
where $\mathbf{\varTheta}_{\text{\text{o,u}}}[n]={\mathbf{U}_{\text{BS}}^{{H}}[n] \mathbf{H}_{\text{\text{u,b}}}[n] \mathbf{U}_{\text{UAV}}[n] \mathbf{U}_{\text{UAV}}^{H}[n] {\mathbf{H}^H_{{\text{u,b}}}[n]} \mathbf{U}_{\text{BS}}[n]}\\\left({\sigma^2\mathbf{U}_{\text{BS}}^{{H}}[n]\mathbf{U}_{\text{BS}}[n]}\right)^{-1}$, with $\mathbf{U}_{\text{UAV}}[n]\in\mathbb{C}^{L\times L}$ being the transmit beamforming matrix generated by the UAV and  $\mathbf{U}_{\text{BS}}[n]\in\mathbb{C}^{M\times M}$ denoting the receive beamforming matrix generated by the BS at the $n$-th time slot.

Considering the downlink period for transmitting the communication results from the BS to GEs via UAV, we assume each GE applies power splitting (PS) protocol to coordinate
the processes of information decoding and energy harvesting from the received signal relayed by the UAV \cite{shi2014joint}.
The received signal at GE $k$ is split to the information decoder (ID) and the energy harvester (EH) by a power splitter. Define $\rho_k[n]$ as the portion of the signal power to the ID, while the remaining portion of power to the EH. Therefore, the SINR of and harvested energy of GE $k$ at time slot $n$ are given by
\vspace{-2mm}
\begin{align}
&\hspace{-2mm}r_{\text{u},k}[n]=\frac{\rho _k\left[ n \right]   \left| \mathbf{{h}}_{k,\text{u}}^H[n] \mathbf{w}_k\left[ n \right] \right|^2}{\rho _k\left[ n \right]\Big(\sum_{j\ne k}{\left| \mathbf{{h}}^H_{k,\text{u}}[n] \mathbf{w}_j\left[ n \right] \right|^2}+\sigma _{k}^{2}\Big) +\delta _{k}^{2} },\\
&\hspace{-2mm}{E^{\text{har}}_k[n]=t_{\text{\text{d}}}}\zeta _k\left( 1-\rho _k\left[ n \right] \right) \Big( \sum_{j=1}^K{\left| \mathbf{h}_{k,\text{u}}^{H}\left[ n \right] \mathbf{w}_j\left[ n \right] \right|^2+}\sigma _{k}^{2} \Big),
\end{align}
where $\mathbf{w}_k\left[ n \right]\in\mathbb{C}^{L\times1}$ denotes the transmit beamforming of the UAV for GE $k$ and {$t_{\text{\text{d}}}\in[0,\delta]$ indicates the predetermined time for the downloading period in each time slot}. $\sigma _{k}^{2}$ is the noise power at GE $k$, while $\delta _{k}^{2}$ represents the additional noise power introduced by the ID at GE $k$. Besides, $0<\zeta _k \leq 1$ denotes the energy conversion efficiency at the EH of GE $k$.

Let $L_{\text{c},k}\left[ n \right]$ and $L_{\text{o},k}\left[ n \right]$ respectively represent the  local computing and the offloaded task bits at time slot $n$. We assume that each GE has a specific computing task  bits to be handled in each time slot, denoted as $\varGamma$. Thus, we have the following task requirement constraints:
\begin{align}
    &{L_{\text{c},k}\left[ n \right] +L_{\text{o},k}\left[ n \right] \ge \varGamma},~{\forall}k, {\forall}n. \label{causal1}
\end{align}

Denote the maximum CPU frequency of GE $k$ as $F^{\text{max}}_k$, then we have the following local computing resource constraints:
\begin{align}
    &L_{\text{c},k}[n]\le {\delta F_{k}^{\rm{max}}}/{C_k},~{\forall}k,~{\forall}n,\label{causal2}
\end{align}
where ${C_k}$ is the number of required CPU cycles for computing one task bit at GE $k$.
Based on \cite{hu2019uav}, the energy consumption of GE $k$ for local computing can be expressed as
\begin{align}
    &E_k^{\text{comp}}[n]={L_{\text{c},k}^{3}\left[ n \right] C_{k}^{3}\varsigma _k}/{\delta ^2},~{\forall}k,~{\forall}n,
\end{align}
where $\varsigma _k$ is the effective capacitance coefficient of GE $k$.

Let $L_{\text{o,u}}\left[ n \right]$ denote the task bits that the UAV further offload to the BS for processing at time slot $n$. 
In this paper, we assume that the computing time at the BS and the transmission time from the BS to the UAV are negligible. We have the following causal constraints for the offloading process:
\begin{align}
    &L_{\text{o},k}[n]\le t_{\text{o}}[n]B\log _2\left( 1+r_k\left[ n \right] \right),~{\forall}k, {\forall}n,\label{causal3}\\
    &L_{\text{o,u}}\left[ n \right] \le {t_{\text{\text{d}}}}R_{{\text{o,u}}}\left[ {n} \right],~{\forall}n,\label{causal4}\\
    &{\sum_{k=1}^K L_{\text{o},k}[n]\le L_{\text{o,u}}[n]},~{\forall}n, \label{causal5}\\
    &\theta L_{\text{o},k}[n]\leq {t_{\text{\text{d}}}}B\log_2(1+r_{\text{u},k}[n]),~{\forall}k, {\forall}n,\label{causal6}
\end{align}
where $\theta$ represents the uniform ratio of the calculation results to the computation tasks.

We introduce an auxiliary variable $\eta$ to denote the minimum remaining energy among all GEs as  shown in constraint \eqref {eta}. Hence, the problem for maximizing $\eta$ can be formulated as
\begin{subequations}
\begin{align}
\textbf{(P1)}~\mathop{\text {max}}\limits_{\mathbf{\Psi}}~&\eta\\
\text{s.t.}~~&(\ref{u_move1}),(\ref{u_move2}),(\ref{causal1}),(\ref{causal2}),(\ref{causal3})-(\ref{causal6}),\\
&\eta \le\sum_{n=1}^NE^{\text{\text{\text{\text{har}}}}}_k[n]- E_k^{\text{comp}}[n] - E_k[n],~{\forall}k,\label{eta}\\
&{t_{\text{o}}[n]+t_{\text{u}}[n] \le \delta}-t_{\text{\text{d}}}, {\forall}n, \label{time}\\
&\text{tr}(\mathbf{U}_{\text{UAV}}\mathbf{U}^H_{\text{UAV}})\le P^{\max}_{\text{UAV}},~{\forall}n, \label{pu2}\\
&\sum_{j=1}^K{\left| \mathbf{w}_j^{H}\left[ n \right] \mathbf{w}_j\left[ n \right] \right|^2}\leq P^{\max}_{\text{UAV}},~{\forall}n,\label{pu1}\\
&0\le E_k[n] \le P_k^{\rm{max}} t_{\text{o}}[n],~{\forall}n,{\forall}k, \label{pk}\\
&0\le \rho _k\left[ n \right] \le 1\label{rho},~{\forall}k,{\forall}n.
\end{align}
where constraints in (\ref{time}) ensure that the time allocated for uplink and downlink periods does not exceed the duration of each time slot. Additionally, constraints (\ref{pu2}), (\ref{pu1}) represent the  power constraints of the UAV for uplink and downlink transmissions, while \eqref{pk} is offloading power constraint for GE $k$, where $ P^{\max}_{\text{UAV}}$ and $ P_k^{\rm{max}}$ are the maximum transmitting power of the UAV and the GE $k$, respectively. In addition, $\mathbf{\Psi}=\left\{ \mathbf{v}_k\left[ n \right], \mathbf{U}_{\text{UAV}}\left[ n \right], \mathbf{U}_{\text{BS}}\left[ n \right], \eta, t_{\text{o}}\left[ n \right], t_{\text{u}}\left[ n \right], E_k\left[ n \right], L_{\text{c,}k}\left[ n \right], \right.\\ \left.L_{\text{o,}k}\left[ n \right], \rho _k\left[ n \right], \mathbf{w}_k\left[ n \right], \mathbf{q}\left[ n \right] \right\}_{k\in\mathcal{K},n\in\mathcal{N}}
$ denotes the compact set of the optimization variables.
\end{subequations}

\vspace{-2mm}
\section{OPTIMIZATION ALGORITHM DESIGN}
In this section, we propose an alternating optimization algorithm to solve the  problem (P1). We divide the optimization variables into four blocks, i.e., {the uplink-period beamforming design set $\mathbf{\Psi}_1= \left\{ \mathbf{v}_{k}\left[ n \right], \mathbf{U}_{\text{UAV}}\left[ n \right], \mathbf{U}_{\text{BS}}\left[ n \right] \right\}$},
the resource allocation set $\mathbf{\Psi}_2= \left\{ \eta,t_{\text{o}}[n], t_{\text{u}}[n], E_{k}[n], L_{\text{c},k}[n], L_{\text{o},k}[n] \right\}$, {the downlink-period beamforming and GEs' PS design set $\mathbf{\Psi}_3= \left\{\eta, \rho_k[n], \mathbf{w}_k\left[ n \right] \right\}$}
and UAV trajectory design set $\mathbf{\Psi}_4= \left\{\eta, \mathbf{q}[n] \right\}$. Therefore, we decompose (P1) into the following four subproblems, which are analyzed and solved as follows.
\subsubsection{{Subproblem for Optimizing the Uplink-Period Beamforming Design Set $\mathbf{\Psi}_1$}}
We employ the zero-forcing (ZF) algorithm for obtaining $\mathbf{v}_{k}[n]$ and the Singular Value Decomposition (SVD)-based approach to analyze the transmission rate from the UAV to the BS. Based on \cite{shi2014joint}, \cite{goldsmith2005wireless}, we can derive the beamforming solutions as:
\begin{align}
    &\mathbf{v}_{k}[n]={\mathbf{\Upsilon}_k[n]\mathbf{\Upsilon}^H_k[n] \overline{\mathbf{h}}_{k,\text{u}}[n]}/{\lVert \mathbf{\Upsilon}_k[n]\mathbf{\Upsilon}^H_k[n]\overline{\mathbf{h}}_{k,\text{u}}\left[ n \right] \rVert},\\
    &\mathbf{U}_{\text{BS}}[n] = [\overline{\boldsymbol{\xi}}_1, \ldots, \overline{\boldsymbol{\xi}}_M],~
    \mathbf{U}_{\text{UAV}}[n] = [\widehat{\boldsymbol{\xi}}_1, \ldots, \widehat{\boldsymbol{\xi}}_L],
\end{align}
where $\mathbf{\Upsilon}_k[n]$ denotes the orthogonal basis for the null space of $\overline{\mathbf{H}}^H_{k,\text{u}}[n]=[\overline{\mathbf{h}}_{1,\text{u}}\left[ n \right],...,\overline{\mathbf{h}}_{k-1,\text{u}}\left[ n \right],\overline{\mathbf{h}}_{k+1,\text{u}}\left[ n \right],...,\overline{\mathbf{h}}_{K,\text{u}}\left[ n \right]]^H$. Also, $\overline{\boldsymbol{\xi}}_m\in\mathbb{R}^{M\times1}$ and $\widehat{\boldsymbol{\xi}}_l\in\mathbb{R}^{L\times1}$ are the normalized eigenvectors of the $m$-th and $l$-th eigenvalues corresponding to {$\overline{\mathbf{H}}_{\text{\text{u,b}}}\left[ n \right] \overline{\mathbf{H}}_{\text{\text{u,b}}}^{H}\left[ n \right]
$ and $\overline{\mathbf{H}}_{\text{\text{u,b}}}^{H}\left[ n \right] \overline{\mathbf{H}}_{\text{\text{u,b}}}\left[ n \right]$}, respectively.

Thus, the channel between the UAV and BS can be divided into several parallel sub-channels. The transmission rate from the UAV to the BS can be formulated as follows
\begin{align}
    &R_{ {\text{o,u}}}[n]=\sum_{i=1}^{{\tau[n]}}{B{\log _2}\bigg( 1+{\frac{\lambda _i{E_{\text{UAV}}^{i}\left[ n \right]}}{t_{\text{u}}[n] d_{\text{u,b}}^{2}[n] \sigma ^2} } \bigg)  },
\end{align}
{where $\tau[n]$ represents the rank of $\overline{\mathbf{H}}_{\text{\text{u,b}}}\left[ n \right]$, and $\lambda_i$ denotes the square of the $i$-th singular value of $\overline{\mathbf{H}}_{\text{\text{u,b}}}\left[ n \right]$.} In addition, ${E_{\text{UAV}}^{i}\left[ n \right]}$ signifies the transmit energy assigned by the UAV to the $i$-th sub-channel at the $n$-th time slot.

\subsubsection{Subproblem for Optimizing the Resource Allocation Set $\mathbf{\Psi}_2$}
To facilitate the subsequent analysis, we introduce a new variable $L^i_{\text{o,u}}[n]$, indicating the offloaded task bits from UAV to BS using the $i$-th sub-channel at time slot $n$. Additionally, we define a new optimization set for subproblem 2, denoted as
$\mathbf{\Psi}'_2 = \left\{ \mathbf{\Psi}_2, \{L^i_{\text{o,u}}[n], E^i_{\text{UAV}}[n]\}_{\forall i,n} \right\}$.
For any given variable sets $\mathbf{\Psi}_1$, $\mathbf{\Psi}_3$ and $\mathbf{\Psi}_4$, the corresponding subproblem can be expressed as follows:
\begin{subequations}
\begin{align}
&~~\textbf{(P2)}~~~~\mathop{\text {max}}\limits_{\mathbf{\Psi}'_2}~~\eta\\
&~~\text{s.t.}~~
(\ref{causal1}),(\ref{causal2}),(\ref{causal3}),(\ref{causal5}),(\ref{causal6}),(\ref{time}),(\ref{pk}),(\ref{eta}),\\
&~~L_{\text{o,u}}\left[ n \right] \le \sum_{i=1}^{\tau[n]}L^i_{\text{o,u}}[n],~{\forall}n,\\
&~~L^i_{\text{o,u}}[n]\le{t_{\text{u}}[n]B{\log _2}\bigg( 1+ {\frac{\lambda _i{E_{\text{UAV}}^{i}\left[ n \right]}}{t_{\text{u}}[n] d_{\text{u,b}}^{2}[n] \sigma ^2} } \bigg)},~{\forall}n, {\forall}i, \label{causal7}\\
&~~\sum_{i=1}^{\tau \left[ n \right]}{E^i_{\text{UAV}}[n] \le P^{\max}_{\text{UAV}}t_{\text{u}}[n]},~{\forall}n.
\end{align}
\end{subequations}

Since  $f(x,t)=t{\rm{{log}}}(1+x/t)$ is a joint concave function w.r.t. $x$ and $t$ for
case of {$x,t\ge0$} \cite{boyd2004convex}, then the constraints (\ref{causal3}), (\ref{causal6}), and (\ref{causal7}) are convex versus the variables in $\mathbf{\Psi}'_2$. Therefore, problem P2 is a standard convex problem that can be solved by the existing solvers, such as CVX.

\subsubsection{{Subproblem for Optimizing the Downlink-Period Beamforming and GEs' PS Design Set $\mathbf{\Psi}_3$}}
By defining $\mathbf{W}_k[n]=\mathbf{w}_k[n]\mathbf{w}_k^H[n]$, ${\mathbf{H}}_{k,\text{u}}[n]={\mathbf{h}}_{k,\text{u}}[n]{\mathbf{h}}^H_{k,\text{u}}[n]$, and introducing an auxiliary variable ${\widetilde{\rho }_k\left[ n \right]}$, which satisfy $e^{\widetilde{\rho }_k\left[ n \right]} = {\rho }_k\left[ n \right]$. Hence, the constraints (\ref{causal6}),  (\ref{eta}) and (\ref{pu1}) can be respectively re-expressed as follows:
\begin{align}
    &\text{tr}\left( \mathbf{H}_{k,\text{u}}\left[ n \right] \mathbf{W}_k\left[ n \right] \right) \ge \big( 2^{\frac{\theta L_{\text{o},k}\left[ n \right]}{ t_{\text{\text{d}}}B } } -1 \big) \times   \nonumber\\
    &~~\bigg( \sum_{j\ne k}{\text{tr}\left( \mathbf{H}_{k,\text{u}}\left[ n \right] \mathbf{W}_j\left[ n \right] \right) +}\left( \delta _{k}^{2}+\sigma ^2 \right) e^{-\widetilde{\rho }_k\left[ n \right]} \bigg),~\forall n, \label{SDR_L}\\
    &\sum_{n=1}^N{t_{\text{\text{d}}}\zeta _k\big( 1-e^{\widetilde{\rho }_k\left[ n \right]}\big) \bigg( \sum_{j=1}^K{\text{tr}\left( \mathbf{H}_{k,\text{u}}\left[ n \right] \mathbf{W}_j\left[ n \right] \right)}+\sigma _{k}^{2} \bigg)}   \nonumber\\
    &~~-E_{k}^{\text{total}}\left[ n \right]\ge \eta,~\forall k, \label{slack}\\
    &\sum_{j=1}^K{\text{tr}\left( \mathbf{W}_{j}\left[ n \right] \right)}\leq {P_{\mathrm{UAV}}^{\max}},~\forall n, \label{SDR_pu}
\end{align}
where $E_{k}^{\text{total}}\left[ n \right] = E_k^{\text{comp}}[n] + E_k[n]$ denotes the total energy consumption of GE $k$ for computing and offloading  at time slot $n$.
Furthermore, we introduce a slack variable $\varOmega _k\left[ n \right]$ to deal with the coupling relationship between $e^{\widetilde{\rho }_k\left[ n \right]}$ and $ \sum_{j=1}^K{\text{tr}\left( \mathbf{H}_{k,\text{u}}\left[ n \right] \mathbf{W}_j\left[ n \right] \right)}+\sigma _{k}^{2} $. Therefore, the constraint (\ref{slack}) can be further re-expressed as the  form in \eqref{SDR,nonconvex}-\eqref{SDR,slack}:
 \begin{align}
     &\sum_{n=1}^N{t_{\text{\text{d}}}\zeta _k\big( e^{\varOmega _k\left[ n \right]}-e^{\widetilde{\rho }_k\left[ n \right]  +\varOmega _k\left[ n \right]} \big) -E^{\text{total}}_k\left[ n \right]}\ge \eta,~\forall k, \label{SDR,nonconvex}\\
    &e^{\varOmega _k\left[ n \right]}\le \big( \sum_{j=1}^K{\text{tr}\left( \mathbf{H}_{k,\text{u}}\left[ n \right] \mathbf{W}_j\left[ n \right] \right)}+\sigma _{k}^{2} \big),~\forall k, \forall n, \label{SDR,slack}
 \end{align}

Hence, for any given variable sets $\mathbf{\Psi}_1$, $\mathbf{\Psi}_2$ and $\mathbf{\Psi}_4$, the subproblem for solving $\mathbf{\Psi}_3$ can be expressed as follows:
\begin{subequations}
\begin{align}
\textbf{(P3)}~&\mathop{\text {max}}\limits_{ \{\mathbf{W}_k[n],\widetilde{\rho }_k\left[ n \right] ,{\varOmega _k[n]}\}_{\forall k,n},\eta}~~\eta\\
&~~~~~\text{s.t.}~~(\ref{SDR_L}), (\ref{SDR_pu}), (\ref{SDR,nonconvex}), (\ref{SDR,slack}),\\
&~~~~~~~~~~0 \le e^{\widetilde{\rho }_k\left[ n \right]} \le 1,~\forall k, \forall n, \label{SDR,e}\\
&~~~~~~~~~~\mathbf{W}_{k}\left[ n \right]\succeq 0,~\forall k, \forall n, \label{SDR,semidefinite}\\
&~~~~~~~~~~\text{Rank}(\mathbf{W}_{{k}}\left[ n \right])=1 ,~\forall k, \forall n,\label{SDR,rank}
\end{align}
\end{subequations}
which is a non-convexity optimization  because of the constraints (\ref{SDR,nonconvex}) and (\ref{SDR,rank}). Fortunately, $e^{\varOmega _k\left[ n \right]}$ is a convex function with respect to $\varOmega _k\left[ n \right]$, and thus we can obtain its lower bound via its first-order Taylor expansion, which is given by
\begin{align}
&\xi _1\left( \varOmega _{k}\left[ n \right] \right) =e^{\varOmega _{k}^{\left( \text{m} \right)}\left[ n \right]}+e^{\varOmega _{k}^{\left( \text{m} \right)}\left[ n \right]}\big( \varOmega _{k}\left[ n \right] -\varOmega _{k}^{\left( \text{m} \right)}\left[ n \right] \big),
\end{align}
where $\varOmega _{k}^{\left( \text{m} \right)}\left[ n \right]$ is the feasible point of  $\varOmega _{k}\left[ n \right]$ at the $m$-th iteration. Thus, the SDR form of problem (P3) is given by
\begin{subequations}
\begin{align}
&{\textbf{(P3.1)}}~\mathop{\text {max}}\limits_{ \{\mathbf{W}_k[n],\widetilde{\rho }_k\left[ n \right] ,{\varOmega _k[n]}\}_{\forall k,n},\eta}~~\eta\\
&\text{s.t.}~~
(\ref{SDR_L}), (\ref{SDR_pu}), (\ref{SDR,slack}),(\ref{SDR,e}),(\ref{SDR,semidefinite}),\\
&\sum_{n=1}^N{t_{\text{\text{d}}}\zeta _k\big(\xi _1\left( \varOmega _{k}\left[ n \right] \right)-e^{\widetilde{\rho }_k\left[ n \right]  +\varOmega _k\left[ n \right]} \big) -E^{\text{total}}_k\left[ n \right]}\ge \eta,\forall k,
\end{align}

It can be noted that problem {(P3.1)} is a standard convex problem that can be solved by CVX. Additionally, ${\rho }_k\left[ n \right]$ can be obtained by the solution to problem (P3.1) according to $e^{\widetilde{\rho }_k\left[ n \right]} = {\rho }_k\left[ n \right]$. However, the solution to (P3.1) may conflict with constraint (\ref{SDR,rank}). Fortunately, we will provide a method to construct a solution satisfying constraint (\ref{SDR,rank}) based on the solution of (P3.1) in the following Theorem \ref{th:1}.
\end{subequations}
\begin{theorem}\label{th:1}
Suppose that the optimal feasible solution of problem {(P3.1)} are $\mathbf{W}^*_k[n]$ ,$\widetilde{\rho }_{k}^{*}\left[ n \right]
$ and ${\varOmega^*_k[n]}$. There exists $\mathbf{W}^{\star}_k[n]$ satisfying $Rank(\mathbf{W}^{\star}_k[n]) = 1$ and other variables $\widetilde{\rho }_{k}^{*}\left[ n \right] $ and ${\varOmega^*_k[n]}$ are still feasible solutions to the problem {(P3.1)}, and the corresponding $\mathbf{W}^*_k[n]$ is given by
\begin{align}
    \mathbf{W}_{{k}}^{\star}\left[ n \right] =\frac{\mathbf{W}_{{k}}^{*}\left[ n \right] {\mathbf{h}}_{{k,\text{u}}}\left[ n \right] {\mathbf{h}}_{{k,\text{u}}}^{{H}}\left[ n \right] \mathbf{W}_{{k}}^{*}\left[ n \right]}{{\mathbf{h}}_{{k,\text{u}}}^{{H}}\left[ n \right] \mathbf{W}_{{k}}^{*}\left[ n \right] {\mathbf{h}}_{{k,\text{u}}}\left[ n \right]}.  \label{proof1}
\end{align}
\end{theorem}
\begin{proof}
According to ($\ref{proof1}$), ${\mathbf{h}}_{{k,\text{u}}}^{{H}}\left[ n \right] \mathbf{W}_{{k}}^{\star}\left[ n \right] {\mathbf{h}}_{{k,\text{u}}}\left[ n \right] ={\mathbf{h}}_{{k,\text{u}}}^{{H}}\left[ n \right] \mathbf{W}_{{k}}^{*}\left[ n \right] {\mathbf{h}}_{{k,\text{u}}}\left[ n \right]$, and $\text{tr}(\mathbf{W}_{{k}}^{\star}\left[ n \right])=\text{tr}(\mathbf{W}_{{k}}^{*}\left[ n \right])$ always hold, which indicates that $\mathbf{W}_{{k}}^{\star}[n]$, $\rho^*_k[n]$, and ${\varOmega^*_k[n]}$ still are  optimal solutions to (P4). The proof has been completed.
\end{proof}
\subsubsection{Subproblem for Optimizing the UAV Trajectory Design Set $\mathbf{\Psi}_4$}
 For any given variable sets $\mathbf{\Psi}_1$, $\mathbf{\Psi}'_2$ and $\mathbf{\Psi}_3$, the subproblem to solve $\mathbf{\Psi}_4$ can be expressed as follows:
\begin{subequations}
\begin{align}
\textbf{(P4)} &~\mathop{\text{max}}\limits_{\mathbf{\Psi}_4} ~~~\eta \\
\text{s.t.}~\nonumber &\sum_{n=1}^N t_{\text{\text{d}}}(1-\rho_k[n]) \zeta_k \bigg(\sum_{j=1}^K \frac{|\overline{\mathbf{h}}_{k,\text{u}}^{H}[n] \mathbf{w}_j[n]|^2}{d_{k,\text{u}}^{2}[n]} + \sigma_{k}^{2}\bigg)\\
&~~~~~~~~~~ -E_k^{\text{\text{total}}}[n] \ge \eta, ~\forall k,\label{p4,nonconvex} \\
&d_{k,\text{u}}^{2}\left[ n \right] \le \frac{E_k\left[ n \right] \left|\mathbf{v}^{\text{H}}_{k}\left[ n \right] \overline{\mathbf{h}}_{k,\text{u}}\left[ n \right]   \right|^2}
{t_{\text{o}}\left[ n \right] \sigma _{k}^{2}\big( 2^{\frac{L_{\text{o},k}\left[ n \right]}{t_{\text{o}}\left[ n \right] B}}-1 \big) },~\forall k, \forall n,\\
\nonumber
&d_{k,\text{u}}^{2}\left[ n \right] \le \frac{-\rho _k\left[ n \right] | \overline{\mathbf{h}}_{k,\text{u}}^{H}\left[ n \right] \mathbf{w}_k\left[ n \right] |^2}
{ \big(\rho _k\left[ n \right] \sigma _{k}^{2}+\delta _{k}^{2} \big)} + \\
&~~\frac{\rho _k\left[ n \right] | \overline{\mathbf{h}}_{k,\text{u}}^{H}\left[ n \right] \mathbf{w}_k\left[ n \right] |^2}
{\big( 2^{\frac{\theta L_{\text{o},k}\left[ n \right]}{t_{\text{\text{d}}}B}}-1 \big) \left( \rho _k\left[ n \right] \sigma _{k}^{2}+\delta _{k}^{2} \right) },
~\forall k, \forall n,\\
&d_{\text{u,b}}^{2}\left[ n \right] \le E_{i}^{\text{UAV}}\left[ n \right] \lambda _i/\sigma ^2t_{\text{u}}\left[ n \right] \big( 2^{\frac{L_{\text{o,u}}^{i}\left[ n \right]}{t_{\text{u}}\left[ n \right] B}}-1 \big),~\forall n,{\forall}i.
\end{align}
\end{subequations}

The non-convexity of problem (P4) arises from the constraint (\ref{p4,nonconvex}). We will further  employ the fractional programming (FP) theory\cite{shen2018fractional} to solve it. Thus, constraint (\ref{p4,nonconvex}) can be transformed into the following form:
\begin{align}
    \nonumber
    &\hspace{-2mm}\sum_{n=1}^N {t_{\text{\text{d}}}}\left( 1-\rho _k\left[ n \right] \right) \zeta _k\Big( \sum_{j=1}^K
    | \overline{\mathbf{h}}_{k,\text{u}}^{H}\left[ n \right] \mathbf{w}_k\left[ n \right] |^2 \varLambda _k\left[ n \right] + \sigma _{k}^{2} \Big)\\
    &~~~~- E^{\text{\text{total}}}_k[n] \le \eta,~{\forall}k, \label{P4,convex}
\end{align}
where $\varLambda _k\left[ n \right] = 2y_k\left[ n \right] -y_{k}^{2}\left[ n \right] d_{\text{u},k}^{2}\left[ n \right]$ with $y_{k}\left[ n \right]$ being an auxiliary variable. Given the trajectory of the UAV, $\mathbf{q}^{(m)}\left[ n \right]$, at the $m$-th iteration, the optimal $y_{k}\left[ n \right]$ can be updated by $y_{k}\left[ n \right]={1}/{{\lVert \left. \mathbf{q}^{(m)}\left[ n \right] -\mathbf{s}_k \rVert ^2+H^2\right.}}\label{update}$.

It can be noted that problem (P4) {with the constraint \eqref{P4,convex}} is a convex optimization problem now. Therefore, problem (P4) can be solved by utilizing the solvers, e.g., CVX.
\begin{figure*}[htbp]
	\centering
	\begin{minipage}{0.328\linewidth}
		\centering
		\includegraphics[width=1\linewidth]{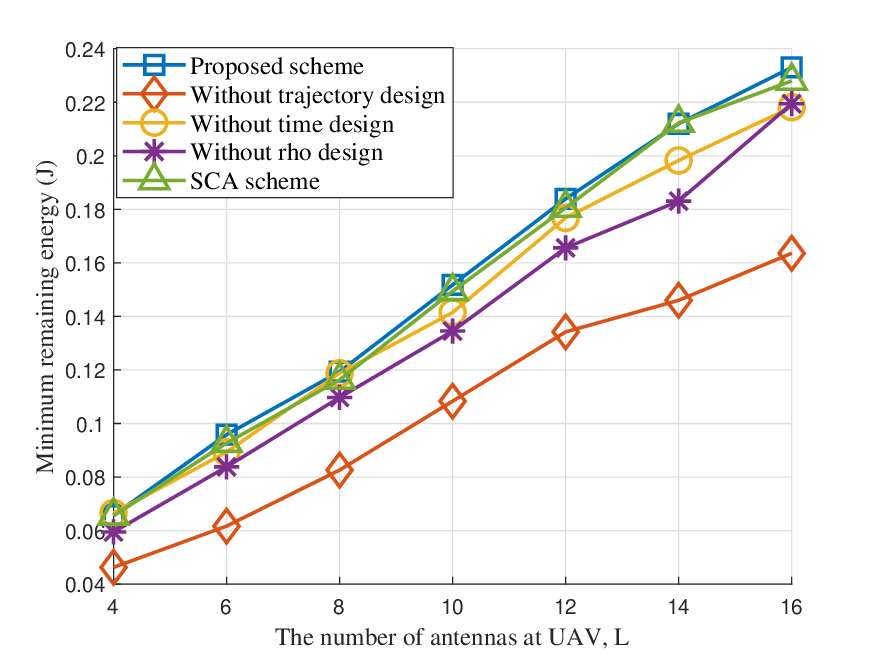}
		\caption{Remaining energy versus the number of antennas at UAV with {${P_{\mathrm{UAV}}^{\max}}$ = 50W}, $H$ = 5m, and, $\Gamma$ = 1Mb.}
		\label{P2}
	\end{minipage}
	\begin{minipage}{0.328\linewidth}
		\centering
		\includegraphics[width=1\linewidth]{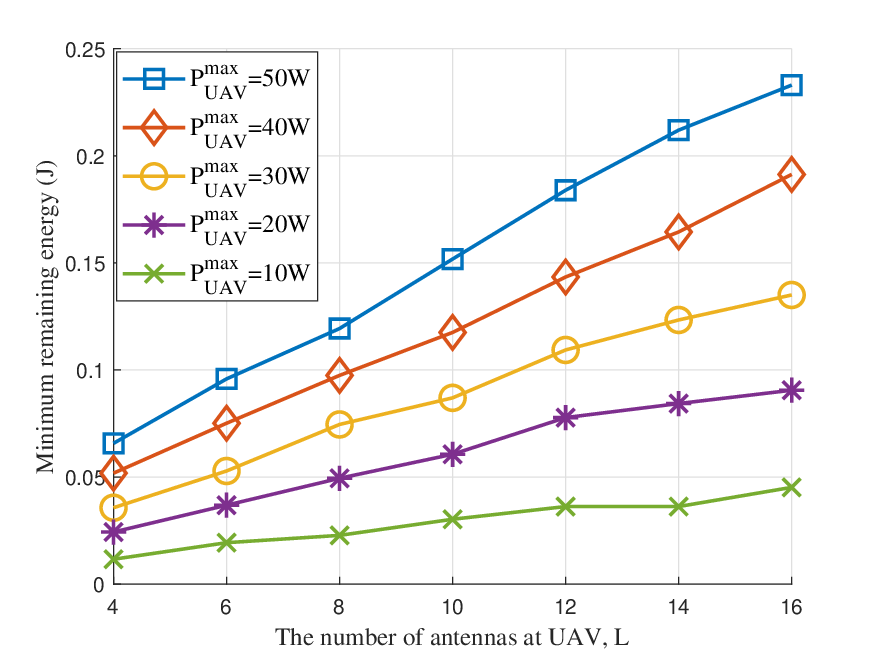}
		\caption{Remaining energy versus transmit power with $H$ = 5m, and, $\Gamma$ = 1Mb.}
		\label{P3}
	\end{minipage}
 	\begin{minipage}{0.328\linewidth}
		\centering
		\includegraphics[width=1\linewidth]{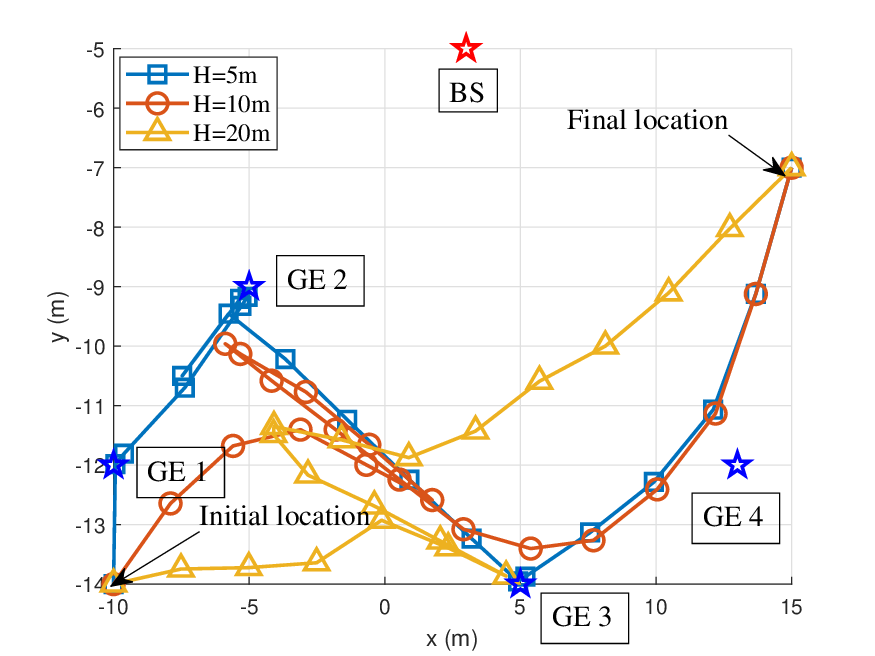}
		\caption{UAV trajectory versus the UAV altitude with $\Gamma$ = 1Mb, $L$ = 8, and, {${P_{\mathrm{UAV}}^{\max}}$ = 50W}.}
		\label{P4}
	\end{minipage}
\end{figure*}
\section{Simulation Results}\label{sec:simulation}
In this section, we simulate the case of $K$= 4 GEs with the coordinates of  (-10, -12), (-5, -9), (5, -14), (13, -12) respectively. Besides, the other simulation parameters are set as $C_k =1000$, $\beta_0$ = - 20 dB, $\sigma^2_k$= - 60 dBm, $\sigma^2$ = - 60 dBm, $\delta^2_k$ = - 50 dBm, $B$ = 10 MHz, $\zeta$ = 10 dB, $\varsigma _k$ = $10^{-28}$, $\theta$ = $10^{-5}$, $F^{\text{max}}_k$ = 2 GHz, $ P_k^{\rm{max}}$ = 1 W, $M_{\text{x}}$ = 4, $M_{\text{y}}$ = 4, $\delta$ = 0.5s, $T$= 10s, $t_{\text{\text{d}}}$ = 0.5$\delta$, $\zeta _k$ = 0.8, $\mathbf{q}_{\text{I}}$= (-10, -14), $\mathbf{q}_{\text{F}}$= (15, -7), $\mathbf{s}_{\text{b}}$ = (3, -5) and $V_{\text{max}}$ = 5m/s.

In Fig. \ref{P2}, the performances of different schemes {versus} the varying numbers of UAV antennas are presented. The scheme without trajectory design refers to fixing the UAV's trajectory as the initial trajectory, while the scheme without time design refers to setting $t_{\text{o}}$ and $t_{\text{u}}$ as 0.25$\delta$. The scheme without rho design refers to setting $\rho_k$ as 0.1, while the SCA scheme refers to optimizing the trajectory using the Successive Convex Approximation (SCA) method, representing a lower bound of the original problem. The performance of all schemes improves as the number of UAV antennas increases, as more antennas provide greater flexibility for beamforming. The proposed scheme is superior to other schemes, demonstrating its effectiveness. The design without trajectory design scheme exhibits inferior performance compared to our proposed scheme, suggesting that modifying the path loss coefficient in the channel through UAV trajectory design can significantly enhance the overall system performance. Furthermore, the scheme without rho design also exhibits a significant performance gap compared to our proposed scheme, which highlights the critical importance of designing the value of $\rho_k$ based on communication requirements in SWIPT networks.

We present the effects of transmit power on performance in Fig. \ref{P3} w.r.t. the number of UAV antennas. At low power levels, the system performance does not significantly improve with the increasing of antennas. However, as the power level increases, the system performance improves more significantly with the increasing number of UAV antennas. Especially when the power is 50W, the performance of the 16-antenna system improves by 254\% compared to the 4-antenna system.

In Fig. \ref{P4}, we compare the UAV trajectory at different altitudes. At an altitude of 5m, the UAV travels to each GE in sequence before flying to the final location. However, at altitudes of 10m or 20m, the UAV's trajectory tends to follow a more central route among GEs. As the altitude increases, the relative difference of distances between the UAV and GEs become smaller, making a more central trajectory more conducive to system performance.

\section{CONCLUSION}
In this paper, we propose a UAV-assisted MEC-SWIPT scheme, which enables the UAV to simultaneously transmit energy and computing results to GEs through the SWIPT technology. Then, we design an alternating optimization algorithm to maximize the minimum remaining energy among all GEs. Simulation results show that the system performance can be significantly enhanced by designing UAV trajectories and GEs' PS ratio for information decoding. The effect of the number of UAV antennas on system performance is also being examined. Additionally, the effectiveness of the proposed scheme is validated by comparing it with the baseline schemes.
\bibliographystyle{IEEEtran}
\bibliography{ref}

\begin{thebibliography}{10}
\providecommand{\url}[1]{#1}
\csname url@samestyle\endcsname
\providecommand{\newblock}{\relax}
\providecommand{\bibinfo}[2]{#2}
\providecommand{\BIBentrySTDinterwordspacing}{\spaceskip=0pt\relax}
\providecommand{\BIBentryALTinterwordstretchfactor}{4}
\providecommand{\BIBentryALTinterwordspacing}{\spaceskip=\fontdimen2\font plus
\BIBentryALTinterwordstretchfactor\fontdimen3\font minus
  \fontdimen4\font\relax}
\providecommand{\BIBforeignlanguage}[2]{{%
\expandafter\ifx\csname l@#1\endcsname\relax
\typeout{** WARNING: IEEEtran.bst: No hyphenation pattern has been}%
\typeout{** loaded for the language `#1'. Using the pattern for}%
\typeout{** the default language instead.}%
\else
\language=\csname l@#1\endcsname
\fi
#2}}
\providecommand{\BIBdecl}{\relax}
\BIBdecl

\bibitem{hu2015mobile}
Y.~C. Hu, M.~Patel, D.~Sabella, N.~Sprecher, and V.~Young, ``Mobile edge
  computing—a key technology towards 5g,'' \emph{ETSI white paper}, vol.~11,
  no.~11, pp. 1--16, 2015.

\bibitem{hu2019uav}
X.~Hu, K.-K. Wong, K.~Yang, and Z.~Zheng, ``{UAV}-assisted relaying and edge
  computing: Scheduling and trajectory optimization,'' \emph{IEEE Transactions
  on Wireless Communications}, vol.~18, no.~10, pp. 4738--4752, 2019.

\bibitem{8805125}
Y.~Du, K.~Yang, K.~Wang, G.~Zhang, Y.~Zhao, and D.~Chen, ``Joint resources and
  workflow scheduling in {UAV}-enabled wirelessly-powered mec for iot
  systems,'' \emph{IEEE Transactions on Vehicular Technology}, vol.~68, no.~10,
  pp. 10\,187--10\,200, 2019.

\bibitem{9312671}
L.~Shi, Y.~Ye, X.~Chu, and G.~Lu, ``Computation energy efficiency maximization
  for a {NOMA}-based {WPT-MEC} network,'' \emph{IEEE Internet of Things
  Journal}, vol.~8, no.~13, pp. 10\,731--10\,744, 2021.

\bibitem{hu2018wireless}
X.~Hu, K.-K. Wong, and K.~Yang, ``Wireless powered cooperation-assisted mobile
  edge computing,'' \emph{IEEE Transactions on Wireless Communications},
  vol.~17, no.~4, pp. 2375--2388, 2018.

\bibitem{xu2022computation}
Y.~Xu, T.~Zhang, Y.~Liu, D.~Yang, L.~Xiao, and M.~Tao, ``Computation capacity
  enhancement by joint {UAV} and {RIS} design in {IoT},'' \emph{IEEE Internet
  of Things Journal}, vol.~9, no.~20, pp. 20\,590--20\,603, 2022.

\bibitem{shi2014joint}
Q.~Shi, L.~Liu, W.~Xu, and R.~Zhang, ``Joint transmit beamforming and receive
  power splitting for {MISO} {SWIPT} systems,'' \emph{IEEE Transactions on
  Wireless Communications}, vol.~13, no.~6, pp. 3269--3280, 2014.

\bibitem{goldsmith2005wireless}
A.~Goldsmith, \emph{Wireless communications}.\hskip 1em plus 0.5em minus
  0.4em\relax Cambridge university press, 2005.

\bibitem{boyd2004convex}
S.~P. Boyd and L.~Vandenberghe, \emph{Convex optimization}.\hskip 1em plus
  0.5em minus 0.4em\relax Cambridge university press, 2004.

\bibitem{shen2018fractional}
K.~Shen and W.~Yu, ``Fractional programming for communication systems—part
  {I}: Power control and beamforming,'' \emph{IEEE Transactions on Signal
  Processing}, vol.~66, no.~10, pp. 2616--2630, 2018.

\end{thebibliography}

\end{document}